\documentclass[submission,copyright,creativecommons]{eptcs}
\providecommand{\event}{EXPRESS/SOS 2017} 
\usepackage{underscore}           

\usepackage{graphicx}
\usepackage{proof}
\usepackage{listings}
\usepackage{url}

\usepackage{amsthm}
\usepackage{amssymb}
\usepackage{amsmath}

\theoremstyle{plain}
\newtheorem{lemma}{Lemma}
\newtheorem{prop}{Proposition}
\theoremstyle{remark}
\newtheorem{case}{Case}[prop] 

\title{Reversing Imperative Parallel Programs}
\author{
James Hoey \qquad\qquad Irek Ulidowski
\institute{Department of Informatics\\
University of Leicester, UK}
\email{\quad jbh11@leicester.ac.uk \quad\qquad iu3@leicester.ac.uk}
\and
Shoji Yuen \institute{Graduate School of Information Science \\ Nagoya University, Japan}
\email{yuen@is.nagoya-u.ac.jp}
}

\def\titlerunning{Reversing Imperative Parallel Programs}
\def\authorrunning{J. Hoey, I. Ulidowski \& S. Yuen}

\allowdisplaybreaks
\begin{document}
\maketitle

\begin{abstract}
We propose an approach and a subsequent extension for reversing imperative programs. Firstly, we produce both an augmented version and a corresponding inverted version of the original program. Augmentation saves reversal information into an auxiliary data store, maintaining segregation between this and the program state, while never altering the data store in any other way than that of the original program. Inversion uses this information to revert the final program state to the state as it was before execution. We prove that augmentation and inversion work as intended, and illustrate our approach with several examples. We also suggest a modification to our first approach to support non-communicating parallelism. Execution interleaving introduces a number of challenges, each of which our extended approach considers. We define annotation and redefine inversion to use a sequence of statement identifiers, making the interleaving order deterministic in reverse.
\end{abstract}

\section{Introduction}
Reverse computation has been an active research area for a number of years. The ability to reverse execute, or invert, a program is desirable due to its potential applications. The relationship with the Landauer principle shows reverse computation to be a feasible solution for producing low power, energy efficient computation \cite{RL1961}. In this paper, we consider reverse computation within the setting of imperative programs. We first propose a state-saving approach for reversing such programs consisting of assignments, conditional statements and while loops. We display an example of this approach showing the execution can now be reversed, and we verify that this reversal is correct. Secondly, we discuss the challenges faced when introducing parallelism, as well as the required modifications to our first approach in order to support it. The formal definition and accompanying example demonstrate the reversal of a parallel program. Finally, we present correctness results for this modified approach. 

The most obvious approach to implementing program reversal is to record the entire program state before executing the program. Recording all of the initial variable values does allow immediate reversal to the original state, however suffers several setbacks, including not re-creating the intermediate program states, and the production of garbage data. We propose an approach that records the necessary information to reverse an execution step-by-step, re-creating intermediate steps faithfully allowing movement in both directions at any point. Any information we save as a result of this is used during inversion, meaning no garbage data is produced.

Inspired by the Reverse C Compiler (RCC) \cite{KP2014,CC1999}, our initial approach takes an original program and produces two versions. The first is the \emph{augmented version}, which becomes the program used for forward execution, and has the capability to save all information necessary for inversion, termed \emph{reversal information}. This is implemented via the function $aug$ that analyses the original program statement by statement, producing the augmented version. The execution of this version populates a collection of initially empty stacks, termed an auxiliary store $\delta$, with this reversal information. Consider the program shown in Figure \ref{original-p}, producing the \texttt{Nth} element of a Fibonacci-like sequence beginning with the values of \texttt{X} and \texttt{Y}. Let the initial state $\sigma$ consist of \texttt{X=4}, \texttt{Y=3}, \texttt{Z=0}, \texttt{N=5} and the initial auxiliary store $\delta$ consist of empty stacks. The execution of the augmented version (displayed later in Section \ref{sec:aug}, Figure \ref{aug-p}) under these stores results in the state $\sigma'$ where \texttt{X=11}, \texttt{Y=18}, \texttt{Z=7}, \texttt{N=2} and auxiliary store $\delta'$ containing reversal information detailed in Section \ref{sec:aug}. With this version now being used for forwards execution, it is crucial that the behaviour with respect to the program state is unchanged. Our first result ensures that if the state $\sigma'$ is produced via the original execution, then it must also be produced via the augmented execution.
 
\lstset{
numbers=left, 
numberstyle=\small, 
numbersep=2pt,
language=Pascal,
framexleftmargin=1pt}
\begin{figure}[t]
  \begin{minipage}[b]{0.49\linewidth}
   \centering
   {\small \begin{lstlisting}[xleftmargin=5.0ex,mathescape=true]
  $\textbf{if}$ $\texttt{X}$ > $\texttt{Y}$ $\textbf{then}$
    $\texttt{Z}$ = $\texttt{Y}$;
    $\texttt{Y}$ = $\texttt{X}$;
    $\texttt{X}$ = $\texttt{Z}$;
  $\textbf{else}$ 
    $\texttt{skip}$
  $\textbf{end}$

  $\textbf{while}$ $\texttt{N}$-2 > $\texttt{0}$ $\textbf{do}$
    $\texttt{Z}$ = $\texttt{X}$;
    $\texttt{X}$ = $\texttt{Y}$;
    $\texttt{Y}$ += $\texttt{Z}$;
    $\texttt{N}$ -= 1;
  $\textbf{end}$
\end{lstlisting} }
    \caption{Original program}
	\label{original-p}
  \end{minipage}
  \hspace{0.01cm}
  \begin{minipage}[b]{0.49\linewidth}
    \centering
    {\small \begin{lstlisting}[xleftmargin=5.0ex,mathescape=true]
  $\textbf{while}$ $\texttt{pop}$($\delta$($\texttt{W}$)) $\texttt{do}$
    $\texttt{N}$ +=1;
    $\texttt{Y}$ -= $\texttt{Z}$;
    $\texttt{X}$ = $\texttt{pop}$($\delta$($\texttt{X}$));
    $\texttt{Z}$ = $\texttt{pop}$($\delta$(Z));
  $\texttt{end}$

  $\textbf{if}$ $\texttt{pop}$($\delta$($\texttt{B}$)) $\textbf{then}$
    $\texttt{X}$ = $\texttt{pop}$($\delta$($\texttt{X}$));
    $\texttt{Y}$ = $\texttt{pop}$($\delta$($\texttt{Y}$));
    $\texttt{Z}$ = $\texttt{pop}$($\delta$($\texttt{Z}$));	
  $\textbf{else}$
    $\texttt{skip}$
  $\textbf{end}$
	\end{lstlisting} }
    \caption{Inverted program}
    \label{inv-p}
  \end{minipage}
\end{figure}

The second version, termed the \emph{inverted version}, is generated via the function $inv$. This version follows the inverted execution order of the original, containing a statement corresponding to each of those of the original. Each inverted statement will typically use information from the auxiliary store to revert all of the effects caused via execution of the original. Consider again the example in Figure \ref{original-p}. Application of $inv$ to this program produces the inverted version, shown in Figure \ref{inv-p}. Execution of this program under the stores $\sigma'$ and $\delta'$ produces the state $\sigma''$ where \texttt{X=4}, \texttt{Y=3}, \texttt{Z=0}, \texttt{N=5}, and the auxiliary store $\delta''$ containing only empty stacks. Our second result validates that $\sigma''$ = $\sigma$ and $\delta''$ = $\delta$, meaning the inversion has happened correctly and the initial program state has been restored. Doing so proves that the augmented program saves the required information, that the inverted program is capable of using this to restore the program state to exactly as it was before, and that augmentation produces no garbage data.  
 
In the second part of the paper, we define a modified approach, this time capable of supporting non-communicating parallelism \cite{HH2010}. Issues introduced such as a non-deterministic execution order make our previous approach insufficient without further state-saving. The \emph{interleaving order}, or order in which the statements are executed, now forms part of the reversal information, captured and stored at runtime. Storing this interleaving order makes the program deterministic in reverse, guaranteeing the execution of the inverted program always follows exactly the inverse execution order of the original. Modifications are made to the process of augmentation from our first approach, with all state-saving implemented at runtime via a set of modified operational semantics for forward execution. A similar reasoning is applied to the process of inversion, resulting in a modified set of operational semantics for reverse execution. These semantics are responsible for using the reversal information, including the interleaving order, to reverse the statements of the original program in exactly the opposite order. Finally the correctness results of our second approach are presented. 

The paper is organised as follows. Section \ref{sec:prog} introduces the programming language and its notion of program state, with the operational semantics given in Section \ref{sec:sos}. Section \ref{sec:aug} describes the process of augmentation and the information that must be saved, as well as proving the correctness of this augmentation. Section \ref{sec:inv} defines the process of inversion and again proves the correctness of this process. Section \ref{sec:update} introduces an updated approach capable of supporting parallel composition, as well as presenting the correctness results.  

\subsection{Related Work} \label{sec:relW}
Program inversion has been discussed for many years, including the work by Gries \cite{DG1981} and by Gl{\"{u}}ck and Kawabe \cite{RG2004,RG2005}. The Reverse C Compiler as described by Perumalla et al. \cite{KP2014,CC1999} is one example of a state-saving approach for the reversal of C programs. We relate very closely to this approach, but with differences including that we currently support a smaller language, and we record a while loop sequence in order to avoid modifying the behaviour of the original program (see Section \ref{sec:aug}). To the best of our knowledge, there is no formal proof of correctness of RCC, and so this is a major focus of our work. Our approach proposes the foundation from which a formally proved approach for a more complex language could emerge.  Other work has been produced on reverse computation used within Parallel Discrete Event Simulation (PDES), a simulation methodology capable of executing events speculatively \cite{RJ1990,DC2016}. The backstroke framework \cite{GV2011} and subsequent work on it by Schordan et al. \cite{MS2015,MS2016} relates slightly less closely to our work as it focuses on this application to PDES. Backstroke is capable of both a state-saving approach and a more advanced, path regeneration method for reverse computation. Other applications include to debugging, with examples being \cite{BB1999,HA1991}. Similarly to program inversion, the reversible programming language Janus, originally proposed in \cite{CL1986} requires additional information within the source code. Any program written in Janus is fully reversible, without the requirement for any control information to be recorded, but with a requirement for additional assertions that make the program deterministic in both directions \cite{TY2007,TY2008A}.  

\section{Programming Language and Program State} \label{sec:prog}
The programming language used for our first approach is similar to any \emph{while language}, particularly that of H\"{u}ttel \cite{HH2010}. This consists of destructive and constructive assignments, with the expression not containing the variable in question, or any side effects. Conditional statements and loops are also supported, implemented using both arithmetic and Boolean expressions.  Let the set of variables $\mathbb{V}$ be ranged over by \texttt{X}, \texttt{Y}, \texttt{Z} \ldots, the set of integers $\mathbb{Z}$ be ranged over by \texttt{l}, \texttt{m} and \texttt{n} and the set of Boolean values $\mathbb{B}$ be \{\texttt{T,F}\}. Also let \texttt{Cop} be the set of constructive assignment operators \{\texttt{+=,-=}\} with \texttt{cop} $\in$ \texttt{Cop}, and \texttt{Op} be the set of arithmetic operators \{\texttt{+,-}\} with \texttt{op} $\in$ \texttt{Op}.
\begin{alignat*}{2}
&\texttt{P}    &\hspace{.2cm}::=   \hspace{.4cm}  &\varepsilon ~|~ \texttt{S};~ \texttt{P}  \\
&\texttt{S}    &\hspace{.2cm}::=   \hspace{.4cm} &\texttt{skip} ~|~  \texttt{X = Exp} ~|~ \texttt{X Cop Exp} ~|~ \texttt{if B then P else P end} ~| \\
&\phantom{\texttt{S}}    &\phantom{::=   \hspace{.4cm}}  &\texttt{while B do P end}  \\
&\texttt{B}    &\hspace{.2cm}::=   \hspace{.4cm}  &\texttt{T} ~|~ \texttt{F} ~|~ \lnot\texttt{B}~|~ \texttt{(B)} ~|~ \texttt{Exp == Exp} ~|~ \texttt{Exp > Exp} ~|~  \texttt{B $\land$ B} \\
&\texttt{Exp}    &\hspace{.2cm}::=   \hspace{.4cm}  &\texttt{X} ~|~ \texttt{n} ~|~ \texttt{(Exp)} ~|~ \texttt{Exp Op Exp}   
\end{alignat*}
$\mathbb{P}$ is the set of programs, ranged over by \texttt{P}, \texttt{Q} and \texttt{R}. $\mathbb{S}$ is the set of statements, ranged over by \texttt{S}. Expressions \texttt{Exp} are ranged over by \texttt{a}, \texttt{a$_0$}, \texttt{a$_0'$}, \texttt{a$_1$}, \texttt{a$_1'$}, Boolean expressions \texttt{B} are ranged over by \texttt{b}, \texttt{b$'$}, \texttt{b$_0$}, \texttt{b$_1$} and expressions that can be either are ranged over by \texttt{ba}, \texttt{ba$_0$}, \texttt{ba$_0'$}, \texttt{ba$_1$}, \texttt{ba$_1'$}. 

The program state is represented via a data store $\sigma$, responsible for mapping each variable to the value it currently holds. A data store is represented as a set of pairs, with the first element of the pair being the variable name, and the second being its current value. A data store is represented formally as the partial function $ \sigma: \mathbb{V} \to \mathbb{Z} $. 

Such stores are manipulated using the following notation. Assuming \texttt{v} $\in \mathbb{Z}$, $\sigma$(\texttt{X}) returns the value currently associated to the variable \texttt{X}, while $\sigma$[\texttt{X} $\mapsto$ \texttt{v}] produces a store identical to $\sigma$, but with the variable \texttt{X} now holding the value \texttt{v}. 

Consider the store $\sigma$, consisting of two variables \texttt{X} and \texttt{Y}, with values \texttt{3} and \texttt{5} respectively, described as $\sigma$ = \{(\texttt{X}, \texttt{3}), (\texttt{Y},\texttt{5})\}. The statement $\sigma$(\texttt{X}) returns \texttt{3}, and $\sigma$[\texttt{X} $\mapsto$ \texttt{10}] results in the store \{(\texttt{X},\texttt{10}), (\texttt{Y},\texttt{5})\}.

\section{Structured Operational Semantics} \label{sec:sos}
This section defines the Structured Operational Semantics (SOS) of the programming language described above. These are defined in the traditional way, following closely with those of H{\"u}ttel \cite{HH2010}. The parameter $\delta$, representing the auxiliary store, is not strictly necessary at this point, but is required later and included here for consistency.
\subsection{Arithmetic Statements}
Let \texttt{v} $\in \mathbb{Z}$ and recall \texttt{op} $\in$ \texttt{Op}.
{ \small \begin{align*}
&\frac{}{(\texttt{X},\sigma,\delta) \rightarrow (\sigma(\texttt{X}),\sigma,\delta)} \hspace{.2cm}
\frac{\texttt{v }=\texttt{ n op m}}{(\texttt{n op m},\sigma,\delta) \rightarrow (\texttt{v},\sigma,\delta)} \hspace{.2cm} \frac{}{(\texttt{(v)},\sigma,\delta) \rightarrow (\texttt{v},\sigma,\delta)} \hspace{.2cm} \frac{(\texttt{a$_0$},\sigma,\delta) \rightarrow (\texttt{a$_0'$},\sigma',\delta')}{(\texttt{(a$_0$)},\sigma,\delta) \rightarrow (\texttt{(a$_0'$)},\sigma',\delta')} \\[10pt]
&\frac{(\texttt{a$_0$},\sigma,\delta) \rightarrow (\texttt{a$_0'$},\sigma',\delta')}{(\texttt{a$_0$} \texttt{ op } \texttt{a$_1$},\sigma,\delta) \rightarrow (\texttt{a$_0'$} \texttt{ op } \texttt{a$_1$},\sigma',\delta')} \hspace{.4cm}
\frac{(\texttt{a$_1$},\sigma,\delta) \rightarrow (\texttt{a$_1'$},\sigma',\delta')}{(\texttt{a$_0$ op } \texttt{a$_1$},\sigma,\delta) \rightarrow (\texttt{a$_0$ op } \texttt{a$_1'$},\sigma',\delta')}
\end{align*} }%
\subsection{Boolean Expressions}
Let \texttt{bop} $\in$ \{\texttt{>,==}\} if used between two arithmetic expressions or \texttt{bop} $\in$ \{\texttt{$\land$,==}\} if used between two Boolean expressions. 
{ \small \begin{align*}
&\frac{}{(\neg \texttt{T},\sigma,\delta) \rightarrow (\texttt{F},\sigma,\delta)} \hspace{.2cm}  \frac{}{(\neg \texttt{F},\sigma,\delta) \rightarrow (\texttt{T},\sigma,\delta)} \hspace{.2cm} \frac{(\texttt{b},\sigma,\delta) \rightarrow (\texttt{b$'$},\sigma',\delta')}{(\neg \texttt{b},\sigma,\delta) \rightarrow (\neg \texttt{b$'$}, \sigma',\delta')} \hspace{.2cm}
\frac{\texttt{ba$_2$} = \texttt{ba$_0$ bop ba$_1$}}{(\texttt{ba$_0$} \texttt{ bop } \texttt{ba$_1$},\sigma,\delta) \rightarrow (\texttt{ba$_2$},\sigma,\delta)} \\[10pt]
&\frac{(\texttt{ba$_0$},\sigma,\delta) \rightarrow (\texttt{ba$_0'$},\sigma',\delta')}{(\texttt{ba$_0$} \texttt{ bop } \texttt{ba$_1$},\sigma,\delta) \rightarrow (\texttt{ba$_0'$} \texttt{ bop } \texttt{ba$_1$},\sigma',\delta')} \hspace{.5cm}
\frac{(\texttt{ba$_1$},\sigma,\delta) \rightarrow ( \texttt{ba$_1'$},\sigma',\delta')}{(\texttt{ba$_0$} \texttt{ bop } \texttt{ba$_1$},\sigma,\delta) \rightarrow (\texttt{ba$_0$} \texttt{ bop } \texttt{ba$_1'$},\sigma',\delta')}
\end{align*} }%
\subsection{Program Statements}
Let \texttt{v} $\in \mathbb{Z}$ and recall \texttt{cop} $\in$ \texttt{Cop}. Let \texttt{op} be \texttt{+} if \texttt{cop} = \texttt{+=}, otherwise let \texttt{op} be \texttt{-}.
{\small \begin{align*}
&\text{[Skip]} \hspace{.2cm} \frac{}{(\texttt{skip;P},\sigma,\delta) \rightarrow (\texttt{P},\sigma,\delta)} \hspace{.8cm} \text{[Seq]} \hspace{.2cm}  \frac{(\texttt{S},\sigma,\delta) \rightarrow (\texttt{S$'$},\sigma',\delta')}{(\texttt{S;P},\sigma,\delta) \rightarrow (\texttt{S$'$;P},\sigma',\delta')} \\[8pt]
&\text{[DA1]} \hspace{.1cm} \frac{}{(\texttt{X = v},\sigma,\delta) \rightarrow (\texttt{skip}, \sigma[\texttt{X} \mapsto \texttt{v}],\delta)} \hspace{.8cm} \text{[DA2]} \hspace{.1cm} \frac{(\texttt{a},\sigma,\delta) \rightarrow (\texttt{a$'$},\sigma',\delta')}{(\texttt{X = a},\sigma,\delta) \rightarrow (\texttt{X = a$'$},\sigma',\delta')} \\[8pt]
&\text{[CA1]} \hspace{.1cm} \frac{}{(\texttt{X cop v}, \sigma,\delta) \rightarrow (\texttt{skip},\sigma[\texttt{X} \mapsto \sigma(\texttt{X})\texttt{ op v}],\delta)} \hspace{.8cm} \text{[CA2]} \hspace{.1cm}  \frac{(\texttt{a},\sigma,\delta) \rightarrow (\texttt{a$'$},\sigma',\delta')}{(\texttt{X cop a},\sigma,\delta) \rightarrow (\texttt{X cop a$'$},\sigma',\delta')}  \\[8pt]
&\text{[C1]} \hspace{.2cm} \frac{}{(\texttt{if T then P else Q end},\sigma,\delta) \rightarrow (\texttt{P},\sigma,\delta)} \hspace{.4cm} \text{[C2]} \hspace{.2cm}  \frac{}{(\texttt{if F then P else Q end},\sigma,\delta) \rightarrow (\texttt{Q},\sigma,\delta)} \\[8pt]
&\text{[C3]} \hspace{.2cm} \frac{(\texttt{b},\sigma,\delta) \rightarrow (\texttt{b$'$},\sigma',\delta')}{(\texttt{if b then P else Q end},\sigma,\delta) \rightarrow (\texttt{if b$'$ then P else Q end},\sigma',\delta')} \\[8pt]
&\text{[Wh]} \hspace{.2cm} \frac{}{(\texttt{P},\sigma,\delta) \rightarrow (\texttt{if b then Q;P else skip end},\sigma,\delta)} \hspace{.5cm} \text{where } \texttt{P} = \texttt{while b do Q end}
\end{align*}}%
\section{Augmentation} \label{sec:aug}
The first step of our first approach is to generate the \emph{augmented version} through a process termed augmentation. This process takes each statement of the original program in succession, and returns a semantically equivalent (with respect to the data store) code fragment containing any required state-saving operations. These fragments are then combined to produce the augmented version.  

The information required to be saved depends on the type of statement. Destructive assignments discard the old value of a variable, meaning it must be saved. Constructive assignments do not suffer this problem meaning they are reversible without state-saving. Due to no guarantee that a condition is invariant, conditional statements must save control information indicating which branch was executed. While loops not having a fixed number of iterations means the number of times the loop should be inverted is unknown. Therefore a sequence of Booleans representing the while loop is saved.  

Saving the result of evaluating conditional statements and while loops removes the burden of re-evaluating these expressions during inversion, unlike the reversible programming language Janus that does require this. In an effort to ensure that the state-saving does not affect the behaviour of the program (w.r.t. the data store), all reversal information is stored separately in an auxiliary data store.  

\subsection{Auxiliary Data Store} \label{sec:auxStore}
Recall that $\mathbb{V}$ is the set of program variable names, and now let both \texttt{B} and \texttt{W} be reserved keywords that cannot appear within this set. An auxiliary data store $\delta$ is a set of stacks, consisting of one self-named stack for each program variable within $\mathbb{V}$, one stack \texttt{B} for all conditional statements and one stack \texttt{W} for all while loops. More formally, $\delta: (\mathbb{V} \to \mathbb{X}) \cup (\{\texttt{B,W}\} \cup \mathbb{B'})$, where $\mathbb{X}$ is the set of stacks of integers and $\mathbb{B'}$ is the set of stacks of Booleans. Auxiliary stores will be represented as a set of pairs. Each pair represents a stack, with the first element being the stack name and the second element being the sequence of its elements. The order of this sequence reflects that of the stack, with the left-most element being the head of the stack. Consider a program consisting of one variable \texttt{X} (initially \texttt{1}) destructively assigned twice (to \texttt{3} and \texttt{5}), one conditional statement that evaluates to \texttt{T} and a while loop with one iteration. The final auxiliary store would be \{(\texttt{X},\{\texttt{3,1}\}), (\texttt{B},\{\texttt{T}\}), (\texttt{W},\{\texttt{T,F}\})\}. 

The stacks on $\delta$ are manipulated in the traditional manner \cite{HH2010}, using \texttt{push} and \texttt{pop} operations introduced via augmentation. The notation $\delta$[\texttt{v} $\mapsto$ \texttt{X}] and \texttt{push(v,X)} both represent pushing the value \texttt{v} to the stack \texttt{X}, while $\delta$[\texttt{X}] and \texttt{pop(X)} represent popping the stack \texttt{X}. Further notation includes $\delta$(\texttt{S}) that returns the stack named \texttt{S}, \texttt{v:S} that indicates a stack with head \texttt{v} and tail \texttt{S}, and $\delta$[\texttt{X/X$'$}] that states the stack \texttt{X} is replaced by \texttt{X$'$}. The SOS rules are defined, where \texttt{v} $\in \mathbb{Z} \cup \mathbb{B}$.%

{\small \begin{align*}
&\text{[Pop]} \hspace{.2cm} \frac{\texttt{$\delta$(X) = v:X$'$}}{(\texttt{pop}(\delta(\texttt{X})),\sigma,\delta) \rightarrow (\texttt{v},\sigma,\delta[\texttt{X/X$'$}])} \hspace{.8cm} \text{[Push1]} \hspace{.2cm} \frac{}{(\texttt{push(v,}\delta(\texttt{X})\texttt{)},\sigma,\delta) \rightarrow (\texttt{skip},\sigma,\delta[\texttt{v} \mapsto \texttt{X}])} \\[10pt]
&\text{[Push2]} \hspace{.2cm} \frac{(\texttt{ba},\sigma,\delta) \rightarrow (\texttt{ba$'$},\sigma',\delta')}{(\texttt{push(}\texttt{ba}\texttt{,}\delta(\texttt{X})\texttt{)},\sigma,\delta) \rightarrow (\texttt{push(}\texttt{ba$'$}\texttt{,}\delta(\texttt{X})\texttt{)},\sigma',\delta')} 
\end{align*} }%
We are now ready to introduce the function that performs the augmentation. 

\subsection{Augmentation Function}
Let $\mathbb{\hat{P}}$ be the set of augmented programs. The function $aug: \mathbb{P} \rightarrow \mathbb{\hat{P}}$ takes the original program and recursively applies the function $a: \mathbb{S} \rightarrow \mathbb{\hat{P}}$ to each statement, producing its augmented version.

Destructive assignments are augmented into two statements, one to push the old value of the variable to its self-named stack on $\delta$, and a second to perform the assignment (see \ref{aass}). Constructive assignments are left unchanged due to their reversibility (see \ref{aop}). Conditional statements have each branch recursively augmented, as well as extended with an operation that stores a Boolean indicating whether the true or false branch was executed (see \ref{Per-co}). As such, $aug$ and $a$ are now defined, where \texttt{cop} $\in$ \texttt{Cop}. 
\begin{align}
aug(\varepsilon) &= \varepsilon \label{augbase}\\
aug(\texttt{S;P}) &= a(\texttt{S}); \text{ } aug(\texttt{P})\label{aug-gen} \\ \nonumber \\
a(\texttt{skip}) &= \texttt{skip} \\
\begin{split} a(\texttt{X = a}) &= \texttt{push(}\sigma(\texttt{X}),\delta(\texttt{X})\texttt{);}  ~\texttt{X = a}\end{split} \label{aass}\\
\begin{split} a(\texttt{X cop a}) &= \texttt{X cop a} \end{split} \label{aop}\\ 
\begin{split}a(\texttt{if b then P else Q end}) &=  \texttt{if b then } aug(\texttt{P}); \texttt{ push(}\texttt{T},\delta(\texttt{B})\texttt{)}   \\ &\phantom{== } \texttt{else } aug(\texttt{Q}); \texttt{ push(}\texttt{F},\delta(\texttt{B})\texttt{) } \texttt{end} \end{split} \label{Per-co}
\end{align}

The traditional approach of handling while loops by initialising a counter and incrementing it for each iteration is not used here due to its adverse effects on the behaviour of the program w.r.t. the data store. While loops are instead augmented to save a sequence of Booleans representing its execution. Generating the sequence in the intuitive way (of a \texttt{T} for each iteration and finally an \texttt{F}) and storing this onto a traditional stack will require the sequence to be manipulated before being used. Such manipulation is both difficult, due to ambiguities within such sequences, and avoidable, by storing a usable order to begin with.

The desired order is that of the intuitive approach, but with any opening \texttt{T} switched with its corresponding closing \texttt{F}, while maintaining any nested \texttt{T} elements. This sequence can be generated provided we can distinguish between the first iteration of a loop and any other. The first iteration now requires an \texttt{F}, while any subsequent iteration (including the unsuccessful last iteration) requires a \texttt{T} (see \ref{aug-while-2}). 
\begin{align}
\begin{split} a(\texttt{while b do P end}) &=  \texttt{if b then} \\ &\phantom{==== } \texttt{push(F},\delta(\texttt{W})\texttt{)}; \text{ } aug(\texttt{P}); \\ &\phantom{==== } \texttt{while b do} \\ &\phantom{====== } \texttt{push(T},\delta(\texttt{W})\texttt{)}; \text{ } aug(\texttt{P}) \\ &\phantom{==== } \texttt{end; }  \texttt{push(T},\delta(\texttt{W})\texttt{)} \\ &\phantom{== } \texttt{else } \texttt{push(}\texttt{F},\delta(\texttt{W})\texttt{)} \texttt{ end}  \end{split} \label{aug-while-2}
\end{align}

We now return to our example discussing Figure \ref{original-p}. The augmented version of this program is shown in Figure \ref{aug-p}. The destructive assignment of \texttt{Z} on line 2 of Figure \ref{original-p} corresponds to line 2 of Figure \ref{aug-p}, where the \texttt{push} statement is used to first save the old value. Lines 5 and 7 of Figure \ref{aug-p} contain inserted operations to save the result of evaluating the conditional statement, while lines 10, 16, 22 and 24 are inserted commands to save the sequence of Boolean values representing the execution of the while loop. Execution of this program under the initial stores $\sigma$ = \{(\texttt{X},\texttt{4}), (\texttt{Y},\texttt{3}), (\texttt{Z},\texttt{0}), (\texttt{N},\texttt{5})\} and $\delta$ = \{(\texttt{X},\{\}), (\texttt{Y},\{\}), (\texttt{Z},\{\}), (\texttt{N},\{\}), (\texttt{B},\{\}), (\texttt{W},\{\})\}, produces the final stores $\sigma'$ = \{(\texttt{X},\texttt{11}), (\texttt{Y},\texttt{18}), (\texttt{Z},\texttt{7}), (\texttt{N},\texttt{2})\} and $\delta'$ = \{(\texttt{X},\{\texttt{7,4,3,4}\}), (\texttt{Y},\{\texttt{3}\}), (\texttt{Z},\{\texttt{4,3,3,0}\}), (\texttt{N},\{\}), (\texttt{B},\{\texttt{T}\}), (\texttt{W},\{\texttt{T,T,T,F}\})\}. The two final stores now contain all of the necessary information for reversal. 
\begin{figure}[t]
  \begin{minipage}[b]{0.49\linewidth}
    \centering
  {\small  \begin{lstlisting}[xleftmargin=5.0ex,mathescape=true,firstnumber=1,numbers=left]
 $\textbf{if}$ $\texttt{X}$ > $\texttt{Y}$ $\textbf{then}$
   $\texttt{push}$($\sigma$($\texttt{Z}$),$\delta$($\texttt{Z}$)); $\texttt{Z}$ = $\texttt{Y}$;
   $\texttt{push}$($\sigma$($\texttt{Y}$),$\delta$($\texttt{Y}$)); $\texttt{Y}$ = $\texttt{X}$;
   $\texttt{push}$($\sigma$($\texttt{X}$),$\delta$($\texttt{X}$)); $\texttt{X}$ = $\texttt{Z}$;
   $\texttt{push}$($\texttt{T}$,$\delta$($\texttt{B}$))
 $\textbf{else}$
   $\texttt{skip;push}$($\texttt{F}$,$\delta$($\texttt{B}$))
 $\textbf{end}$
 $\textbf{if}$ $\texttt{N}$-2 > 0 $\textbf{then}$
   $\texttt{push}$($\texttt{F}$,$\delta$($\texttt{W}$));
   $\texttt{push}$($\sigma$($\texttt{Z}$),$\delta$($\texttt{Z}$)); $\texttt{Z}$ = $\texttt{X}$;
   $\texttt{push}$($\sigma$($\texttt{X}$),$\delta$($\texttt{X}$)); $\texttt{X}$ = $\texttt{Y}$;
    \end{lstlisting}  }
  \end{minipage}
  \hspace{0.01cm}
  \begin{minipage}[b]{0.49\linewidth}
    \centering
  {\small  \begin{lstlisting}[xleftmargin=5.0ex,mathescape=true,firstnumber=13,numbers=left]
   $\texttt{Y}$ += $\texttt{Z}$;
   $\texttt{N}$ -= 1;
   $\textbf{while}$ $\texttt{N}$-2 > 0 $\textbf{do}$
      $\texttt{push}$($\texttt{T}$,$\delta$($\texttt{W}$));
      $\texttt{push}$($\sigma$($\texttt{Z}$),$\delta$($\texttt{Z}$)); $\texttt{Z}$ = $\texttt{X}$;
      $\texttt{push}$($\sigma$($\texttt{X}$),$\delta$($\texttt{X}$)); $\texttt{X}$ = $\texttt{Y}$;
      $\texttt{Y}$ += $\texttt{Z}$;
      $\texttt{N}$ -= 1
   $\textbf{end}$
   $\texttt{push}$($\texttt{T}$,$\delta$($\texttt{W}$))
 $\textbf{else}$
   $\texttt{push}$($\texttt{F}$,$\delta$($\texttt{W}$));$\textbf{ end}$
	\end{lstlisting} }
  \end{minipage}
  \caption{Augmented Version of the Program in Figure \ref{original-p}}
  \label{aug-p}
\end{figure}

We are now ready to state our first result. Firstly, Proposition \ref{prop1} states that if the execution of an original program terminates, then the execution of the augmented version of that program also terminates (where a program terminates if its execution finishes with the configuration $(\texttt{skip},\sigma^*,\delta^*)$ for some $\sigma^*$ and $\delta^*$). Secondly, Proposition \ref{prop1} states that augmentation produces an augmented version that modifies the data store $\sigma$ in exactly the same way as that of the original program to $\sigma'$, while also populating the auxiliary store $\delta$ with reversal information producing $\delta'$. 

\begin{prop} \label{prop1}
Let \textup{\texttt{P}} be a program that does not interact with the auxiliary store, $\sigma$ be an arbitrary initial data store and $\delta$ be an arbitrary initial auxiliary data store. Firstly, if $(\textup{\texttt{P}},\sigma,\delta) \rightarrow^* (\textup{\texttt{skip}},\sigma',\delta'')$, for some $\sigma'$ and $\delta''$, then $(aug(\textup{\texttt{P}}),\sigma,\delta) \rightarrow^* (\textup{\texttt{skip}},\sigma'',\delta''')$ for some $\sigma''$ and $\delta'''$. Secondly, if $(\textup{\texttt{P}},\sigma,\delta) \rightarrow^* (\textup{\texttt{skip}},\sigma',\delta)$, for some $\sigma'$, then $(aug(\textup{\texttt{P}}),\sigma,\delta) \rightarrow^* (\textup{\texttt{skip}},\sigma',\delta')$ for some $\delta'$. 
\end{prop}
We note that the inverse implication, namely that if $(\texttt{$aug($P$)$},\sigma,\delta) \rightarrow^* (\texttt{skip},\sigma',\delta')$, for some $\sigma'$ and $\delta'$, then $(\textup{\texttt{P}},\sigma,\delta) \rightarrow^* (\textup{\texttt{skip}},\sigma',\delta)$, would also be valid. However we defer this proof to future work, and now return to proving the second part of Proposition \ref{prop1} (with the first following correspondingly).
\begin{proof} By induction on the length of the sequence (\texttt{P}, $\sigma,\delta$) $\rightarrow^*$ (\texttt{skip}, $\sigma',\delta$). Since there are no transitions of length 0, the proposition holds vacuously. Assume that the proposition holds for programs \texttt{R}, stores $\sigma^*$ and auxiliary stores $\delta^*$, such that $(\texttt{R},\sigma^*,\delta^*) \rightarrow^* (\texttt{skip},\sigma^{*}_1,\delta^*)$ is shorter than $(\texttt{P},\sigma,\delta) \rightarrow^* (\texttt{skip},\sigma',\delta)$. Further assuming \texttt{P} is of the form \texttt{S;P$'$} such that \texttt{S} is a statement and $\texttt{P}'$ is the remaining program, we have that $$(\texttt{S;P$'$},\sigma,\delta) \rightarrow^* (\texttt{skip},\sigma',\delta)$$ for some $\sigma'$. Through use of the SOS rules \texttt{Seq} and \texttt{Skip}, we have $$(\texttt{S;P$'$},\sigma,\delta) \rightarrow^* (\texttt{skip;P$'$},\sigma'',\delta) \rightarrow (\texttt{P$'$},\sigma'',\delta) \rightarrow^* (\texttt{skip},\sigma',\delta)$$ for some $\sigma''$. With this in mind, we need to show that $(aug(\texttt{S;P$'$}),\sigma,\delta) \rightarrow^* (\texttt{skip},\sigma',\delta')$ for some $\delta'$. 

By the definition of $aug$, clause (\ref{aug-gen}) we have $(aug(\texttt{S;P$'$}),\sigma,\delta) = (a(\texttt{S});aug(\texttt{P$'$}),\sigma,\delta)$, meaning it is sufficient to prove $$\begin{array}{l}
(a(\texttt{S});aug(\texttt{P$'$}),\sigma,\delta) \rightarrow^* (\texttt{skip};aug(\texttt{P$'$}),\sigma'',\delta'') \rightarrow (aug(\texttt{P$'$}),\sigma'',\delta'')  \rightarrow^* (\texttt{skip},\sigma',\delta') \end{array}$$ for some $\sigma''$, $\delta''$ and $\delta'$. Since $(\texttt{S;P$'$},\sigma,\delta) \rightarrow^* (\texttt{P$'$},\sigma'',\delta)$, then repeated use of the Seq rule (from conclusion to premises) produces $(\texttt{S},\sigma,\delta) \rightarrow^* (\texttt{skip},\sigma'',\delta)$. Now assume  $a(\texttt{S})$ = \texttt{P$_S$} for each type of statement \texttt{S}. Then by Lemma \ref{lem-2} below, we have that $(\texttt{P$_S$},\sigma,\delta) \rightarrow^* (\texttt{skip},\sigma'',\delta'')$ for some $\delta''$. Using the Seq rule (from premises to conclusion) we obtain $$(\texttt{P$_S$};aug(\texttt{P$'$}),\sigma,\delta) \rightarrow^* (\texttt{skip};aug(\texttt{P$'$}),\sigma'',\delta'')$$ Then by the Skip rule, we get $ (\texttt{skip};aug(\texttt{P$'$}),\sigma'',\delta'') \rightarrow (aug(\texttt{P$'$}),\sigma'',\delta'')$. The induction hypothesis on $(\texttt{P$'$},\sigma'',\delta'')$ gives us $$(aug(\texttt{P$'$}),\sigma'',\delta'') \rightarrow^* (\texttt{skip},\sigma',\delta')$$ for some $\delta'$. Therefore we have obtained $ \begin{array}{l}
(a(\texttt{S});aug(\texttt{P$'$}),\sigma,\delta) \rightarrow^* (aug(\texttt{P$'$}),\sigma'',\delta'') \rightarrow^* (\texttt{skip},\sigma',\delta')
\end{array}$, meaning $(aug(\texttt{S;P$'$}),\sigma,\delta) \rightarrow^* (\texttt{skip},\sigma',\delta')$ holds as required. Therefore the proposition holds, provided the following lemma holds.
\end{proof}
\begin{lemma} \label{lem-2}
Let \texttt{S} be a statement that does not interact with the auxiliary store, $\sigma$ be an initial data store and $\delta$ be an initial auxiliary data store.  If $(\textup{\texttt{S}},\sigma,\delta) \rightarrow^* (\textup{\texttt{skip}},\sigma',\delta)$ for some $\sigma'$ then $(a(\textup{\texttt{S}}),\sigma,\delta) \rightarrow^* (\textup{\texttt{skip}},\sigma',\delta')$ for some $\delta'$. 
\end{lemma}
\begin{proof}
We consider each type of statement \texttt{S} in turn. Due to space constraints, we only include one case, with the other cases following similarly. The notation $(\texttt{Q},\sigma,\delta) \xrightarrow[\text{X}]{}^l (\texttt{Q$'$},\sigma^\dagger,\delta^\dagger)$ denotes $l$ transitions by the SOS rule \text{X} produces the program \texttt{Q$'$}, store $\sigma^\dagger$ and auxiliary store $\delta^\dagger$.

\begin{case} \label{c-da-lem1}
Consider statement \texttt{X=a} and its execution under the initial stores $\sigma$ and $\delta$ where \texttt{X} is initially \texttt{v$'$} and \texttt{a} evaluates to \texttt{v} in $l$ steps such that \[ (\texttt{X=a},\sigma,\delta) \xrightarrow[\text{DA2}]{}^l (\texttt{X=v},\sigma,\delta) \xrightarrow[\text{DA1}]{} (\texttt{skip},\sigma[\texttt{X} \mapsto \texttt{v}], \delta).\]  Recall that $ \sigma(\texttt{X})$ = \texttt{v$'$}. The execution of the augmented version of \texttt{X=a} is $$\begin{array}{l}
(\texttt{push(\texttt{v$'$},$\delta$(X))};\texttt{X=a},\sigma,\delta) \xrightarrow[\text{Push$_1$, Skip}]{} (\texttt{X=a},\sigma,\delta[\texttt{v$'$} \mapsto \texttt{X}]) \\[4pt] \xrightarrow[\text{DA2}]{}^l (\texttt{X=v},\sigma,\delta[\texttt{v$'$} \mapsto \texttt{X}]) \xrightarrow[\text{DA1}]{}  (\texttt{skip},\sigma[\texttt{X} \mapsto \texttt{v}],\delta[\texttt{v$'$} \mapsto \texttt{X}])
\end{array}$$
As such, this case holds with $\sigma'$ = $\sigma[\texttt{X} \mapsto \texttt{v}]$ and $\delta'$ = $\delta[\texttt{v$'$} \mapsto \texttt{X}]$. 
\end{case}
With all other cases following in a similar manner, Lemma \ref{lem-2} holds.  
\end{proof}

\section{Inversion} \label{sec:inv}
The second step of our initial approach is to generate the inverted version through a process named inversion. Inversion takes each statement in reverse order, generates the code fragment necessary to undo its effects, before combining these fragments to generate the inverted version. The majority of the returned code fragments will use the reversal information on the auxiliary store, meaning the augmented version must be executed prior to the execution of this version.

Destructive assignments are replaced with another destructive assignment to the same variable, but this time assigning the value currently at the top of the self-named stack (see \ref{inv:x=a}). Constructive assignments require no reversal information, and can simply be replaced by their inverse (see \ref{ioprule}). Conditional statements saved a Boolean indicating which branch was executed, meaning the retrieval and evaluation of this now replaces the original condition, along with the recursive inversion of the branches (see \ref{iif}). While loops saved a sequence of Booleans in the desired order, meaning the while loop can continually iterate until the top of the stack \texttt{W} is no longer true (see \ref{iwhile}), along with the recursive inversion of the body. As mentioned earlier, the reverse execution of conditionals and loops does not require their conditions to be re-evaluated, increasing efficiency. 

Let $\mathbb{P}^{-1}$ be the set of inverted programs. The function $inv: \mathbb{P} \rightarrow \mathbb{P}^{-1}$ takes the original program and recursively applies the function $i: \mathbb{S} \rightarrow \mathbb{P}^{-1}$ to each statement in reverse order, producing its inverted version. We now define $inv$ and $i$, where \texttt{cop} $\in$ \texttt{Cop}, and \texttt{icop} = \texttt{+=} if \texttt{cop} = \texttt{-=}, and \texttt{-=} otherwise. 
\begin{align}
\textit{inv}(\varepsilon) &= \varepsilon \label{invbase}\\
\textit{inv}(\texttt{S;P}) &= \textit{inv}(\texttt{P}); \textit{ i}(\texttt{S})  \label{invrec}\\ \nonumber \\
i(\texttt{skip}) &= \texttt{skip} \\
\textit{i}(\texttt{X = a}) &= \texttt{X = }\texttt{pop(}\delta(\texttt{X})\texttt{)} \label{inv:x=a}\\
\begin{split} \textit{i}(\texttt{X cop a}) &= \texttt{X icop a} \end{split} \label{ioprule}\\
\begin{split} \textit{i}(\texttt{if b then P else Q end}) &= \texttt{if }\texttt{pop(}\delta(\texttt{B})\texttt{)}\texttt{ then }inv(\texttt{P})  \texttt{ else } inv(\texttt{Q}) \texttt{ end} \end{split} \label{iif}\\
\begin{split} \textit{i}(\texttt{while b do P end}) &=  \texttt{while }\texttt{pop(}\delta(\texttt{W})\texttt{)}\texttt{ do } inv(\texttt{P}) \texttt{ end} \end{split} \label{iwhile}
\end{align}
We now return to our example code shown in Figure \ref{original-p}. Applying the function $inv$ to this program produces the inverted program shown in Figure \ref{inv-p}. The overall program order has been inverted, with the while loop now being executed first. An example destructive assignment is on line 2 of Figure \ref{original-p}, and is inverted via the line 11 of Figure \ref{inv-p}. Execution of this version under the final stores $\sigma'$ = \{(\texttt{X},\texttt{11}), (\texttt{Y},\texttt{18}), (\texttt{Z},\texttt{7}), (\texttt{N},\texttt{2})\} and $\delta'$ = \{(\texttt{X},\{\texttt{7,4,3,4}\}), (\texttt{Y},\{\texttt{3}\}), (\texttt{Z},\{\texttt{4,3,3,0}\}), (\texttt{N},\{\}), (\texttt{B},\{\texttt{T}\}), (\texttt{W},\{\texttt{T,T,T,F}\})\}, produces the initial stores $\sigma''$ = \{(\texttt{X},\texttt{4}), (\texttt{Y},\texttt{3}), (\texttt{Z},\texttt{0}), (\texttt{N},\texttt{5})\} and $\delta''$ = \{(\texttt{X},\{\}), (\texttt{Y},\{\}), (\texttt{Z},\{\}), (\texttt{N},\{\}), (\texttt{B},\{\}), (\texttt{W},\{\})\}. As should be clear, $\sigma''$ = $\sigma$ and $\delta''$ = $\delta$, meaning the reversal has executed successfully. 

We will now present our second result for $(\texttt{P},\sigma,\delta)$. Recall that by Proposition \ref{prop1}, the execution of the augmented version of \texttt{P} produces the modified auxiliary store $\delta'$, which plays a crucial role in Proposition \ref{prop2} below. Firstly, Proposition \ref{prop2} states that if the original program \texttt{P} terminates on $\sigma$ and $\delta$, producing $\sigma'$, then the execution of the inverted version of \texttt{P} terminates on $\sigma'$ and the modified auxiliary store $\delta'$. Secondly, Proposition \ref{prop2} states that given the final stores $\sigma'$ and $\delta'$ produced via execution of the augmented version of \texttt{P}, executing the corresponding inverted version on these stores restores the initial state, namely $\sigma$ and $\delta$.

\begin{prop}\label{prop2}
Let \textup{\texttt{P}} be a program that does not interact with the auxiliary store, $\sigma$ be an arbitrary initial data store and $\delta$ be an arbitrary initial auxiliary data store. Firstly, if $(\textup{\texttt{P}},\sigma,\delta) \rightarrow^* (\textup{\texttt{skip}},\sigma',\delta'')$, for some $\sigma'$ and $\delta''$, then $(inv(\textup{\texttt{P}}),\sigma',\delta') \rightarrow^* (\textup{\texttt{skip}},\sigma'',\delta''')$, for some $\sigma''$ and $\delta'''$. Secondly, if $(\textup{\texttt{P}},\sigma,\delta) \rightarrow^* (\textup{\texttt{skip}},\sigma',\delta)$, for some $\sigma'$, then $(inv(\textup{\texttt{P}}),\sigma',\delta') \rightarrow^* (\textup{\texttt{skip}},\sigma,\delta)$, for some $\delta'$. 
\end{prop}

We note that the first result in Proposition \ref{prop2} would be valid, but postpone the proof to future work and now return to proving the second part of Proposition \ref{prop2}.

\begin{proof}
By induction on the length of the sequence (\texttt{P}, $\sigma,\delta$) $\rightarrow^*$ (\texttt{skip}, $\sigma',\delta$). Since there are no transitions of length 0, the proposition holds vacuously. Assume that the proposition holds for programs \texttt{R}, stores $\sigma^*$ and auxiliary stores $\delta^*$, such that $(\texttt{R},\sigma^*,\delta^*) \rightarrow^* (\texttt{skip},\sigma^{*}_1,\delta^*)$ is shorter than $(\texttt{P},\sigma,\delta) \rightarrow^* (\texttt{skip},\sigma',\delta)$. Further assume \texttt{P} is of the form \texttt{S;P$'$} such that \texttt{S} is a statement and $\texttt{P}'$ is the remaining program. Let $(\texttt{S;P$'$},\sigma,\delta) \rightarrow^* (\texttt{skip},\sigma',\delta)$ for some $\sigma'$. This means that $$ (\texttt{S;P$'$},\sigma,\delta) \rightarrow^* (\texttt{skip;P$'$},\sigma'',\delta) \rightarrow (\texttt{P$'$},\sigma'',\delta) \rightarrow^* (\texttt{skip},\sigma',\delta) $$ for some $\sigma''$. By applying the Seq rule (conclusion to premises) to $(\texttt{S;P$'$},\sigma,\delta) \rightarrow^* (\texttt{skip;P$'$},\sigma'',\delta)$, we obtain $(\texttt{S},\sigma,\delta) \rightarrow^* (\texttt{skip},\sigma'',\delta)$. By the definition of $aug$, clause (\ref{aug-gen}) we have $(aug(\texttt{S;P$'$}),\sigma,\delta) = (a(\texttt{S}); aug(\texttt{P$'$}), \sigma,\delta)$ and by Proposition \ref{prop1} we have $$\begin{array}{l}
(a(\texttt{S});aug(\texttt{P$'$}),\sigma,\delta) \rightarrow^* (\texttt{skip};aug(\texttt{P$'$}),\sigma'',\delta'')  \rightarrow (aug(\texttt{P$'$}),\sigma'',\delta'') \rightarrow^* (\texttt{skip},\sigma',\delta')
\end{array}$$ for some $\delta''$, $\delta'$. Using Seq (conclusion to premise) on $(a(\texttt{S});aug(\texttt{P$'$}),\sigma,\delta)   \rightarrow^* (\texttt{skip};aug(\texttt{P$'$}),\sigma'', \delta'')$ we obtain $(a(\texttt{S}),\sigma,\delta) \rightarrow^* (\texttt{skip},\sigma'',\delta'')$.

We need to show that given $\sigma'$ and $\delta'$, $(inv(\texttt{S;P$'$}),\sigma',\delta') \rightarrow^* (\texttt{skip},\sigma,\delta)$. By the definition of $inv$, clause \ref{invrec}, we have $inv(\texttt{S;P$'$}) = inv(\texttt{P$'$});i(\texttt{S})$, meaning we shall show $$ (inv(\texttt{P$'$});i(\texttt{S}),\sigma',\delta') \rightarrow^* (i(\texttt{S}),\sigma^\dagger,\delta^\dagger) \rightarrow^* (\texttt{skip},\sigma,\delta) $$ for some $\sigma^\dagger$ and $\delta^\dagger$. The induction hypothesis for $(\texttt{P$'$},\sigma'',\delta'') \rightarrow^* (\texttt{skip},\sigma',\delta'')$, where $\delta''$ is obtained by augmentation of \texttt{S} on $\delta$, gives us $(inv(\texttt{P$'$}),\sigma',\delta') \rightarrow^* (\texttt{skip},\sigma'',\delta'')$ where $\delta'$ is obtained by augmentation of \texttt{P$'$} on $\sigma''$ and $\delta''$ as shown by $(aug(\texttt{P}),\sigma'',\delta'') \rightarrow^* (\texttt{skip},\sigma',\delta')$ above. Using the rule Seq (premise to conclusion) repeatedly we get $$ (inv(\texttt{P$'$});i(\texttt{S}),\sigma',\delta') \rightarrow^* (\texttt{skip};i(\texttt{S}),\sigma'',\delta'') \rightarrow (i(\texttt{S}),\sigma'',\delta'') $$ Therefore $\sigma^\dagger = \sigma''$ and $\delta^\dagger = \delta''$. All that remains now is to prove $(i(\texttt{S}),\sigma'',\delta'') \rightarrow^* (\texttt{skip},\sigma,\delta) $, which is done in Lemma \ref{lem-3} below. 
\end{proof}

\begin{lemma} \label{lem-3}
Let \textup{\texttt{S}} be a statement that does not interact with the auxiliary store, $\sigma$ be an arbitrary initial data store and $\delta$ be an arbitrary initial auxiliary data store. Then if $(\textup{\texttt{S}},\sigma,\delta) \rightarrow^* (\textup{\texttt{skip}},\sigma',\delta)$ for some $\sigma'$, then $(i(\textup{\texttt{S}}),\sigma',\delta') \rightarrow^* (\textup{\texttt{skip}},\sigma,\delta)$ for some $\delta'$. 
\end{lemma}
\begin{proof}
We consider each type of statement \texttt{S} in turn. Due to space constraints, we only include one case, with the other cases following similarly.

\begin{case}
Consider statement \texttt{X=a} and its execution under the initial stores $\sigma$ and $\delta$ where \texttt{X} is initially \texttt{v$'$} and \texttt{a} evaluates to \texttt{v} in $l$ steps such that \[ (\texttt{X=a},\sigma,\delta) \xrightarrow[\text{DA2}]{}^l (\texttt{X=v},\sigma,\delta) \xrightarrow[\text{DA1}]{} (\texttt{skip},\sigma[\texttt{X} \mapsto \texttt{v}], \delta)\] Then by Proposition \ref{prop1} and Lemma \ref{lem-2} Case \ref{c-da-lem1}, we have that ($a($\texttt{X=a}$)$,$\sigma$,$\delta$) $\rightarrow \ldots \rightarrow (\texttt{skip},\sigma'$,$\delta'$), such that $\sigma' = \sigma[\texttt{X} \mapsto \texttt{v}]$ and $\delta' = \delta[\texttt{v$'$} \mapsto \texttt{X}]$. Then the execution of the inverted version of \texttt{X=a} is $$\begin{array}{l}
(i(\texttt{X=a}),\sigma',\delta') \xrightarrow[\text{Pop}]{} (\texttt{X=v$'$},\sigma',\delta'[\texttt{X}]) \xrightarrow[\text{DA1}]{}  (\texttt{skip},\sigma'[\texttt{X} \mapsto \texttt{v$'$}]), \delta'[\texttt{X}])
\end{array}$$ such that $\sigma'[\texttt{X} \mapsto \texttt{v$'$}] = \sigma$ since \texttt{v$'$} is equal to the initial value of \texttt{X} retrieved by the \texttt{pop} operation, and $\delta'[\texttt{X}] = \delta$. Therefore the stores have been restored to their initial states, as required, meaning the case holds.
\end{case}
With all other cases following in a similar manner, the Lemma is proved to be correct.
\end{proof}

\section{Adding Parallelism} \label{sec:update}
We will now modifiy our first approach to support non-communicating parallelism \cite{HH2010}, also referred to as \emph{interleaving}, where the execution of two (or more) programs are interleaved while each individually maintains program order. To the best of our knowledge, RCC does not support parallelism in any form. Due to space constraints, we restrict the language to assignments and parallelism only. Conditionals and loops can be modified in a similar way and so are omitted. Let us reuse previous notation such that $\mathbb{P}$ is now the set of programs of this restricted language, $\hat{\mathbb{P}}$ is now the set of annotated programs, $\mathbb{P}^{-1}$ is now the set of inverted programs and $\mathbb{S}$ is now the set of statements of this restricted language. The definition of a statement becomes 
\begin{align*}
\begin{split} \texttt{S} &::= \texttt{skip} ~|~  \texttt{X = Exp} ~|~ \texttt{X Cop Exp} ~|~  \texttt{P par P}  \end{split}
\end{align*}
\subsection{Challenges}
\lstset{
numbers=none, 
numberstyle=\small, 
numbersep=2pt,
language=Pascal, 
framexleftmargin=15pt}
\begin{figure}[t]
  \begin{minipage}[b]{0.49\linewidth}
   \centering
    \begin{lstlisting}[xleftmargin=5.0ex,mathescape=true]
$\texttt{X}$+=$\texttt{Y}$+$\texttt{2}$ $\texttt{par}$ ($\texttt{Y}$=$\texttt{X}$+$\texttt{2}$; $\texttt{X}$=$\texttt{4}$)
\end{lstlisting}
    \caption{Original program}
	\label{simp-ex}
  \end{minipage}
  \hspace{0.01cm}
  \begin{minipage}[b]{0.49\linewidth}
    \centering
    \begin{lstlisting}[xleftmargin=5.0ex,mathescape=true]
$\texttt{X}$+=$\texttt{Y}$+$\texttt{2}$[] $\texttt{par}$ ($\texttt{Y}$=$\texttt{X}$+$\texttt{2}$[]; $\texttt{X}$=$\texttt{4}$[])
	\end{lstlisting}
    \caption{Annotated program}
    \label{simp-ex-ann}
  \end{minipage}
\end{figure}
Supporting parallelism introduces three challenges. Execution of parallel programs results in a non-deterministic execution order. Our first approach works as the programs are sequential, allowing the inverted program to follow the \emph{inverted program order}. However there may be different execution orders of the same parallel program due to interleaving. So, without care, programs can be executed forwards under one interleaving and reversed under another, which is clearly incorrect. Consider the program in Figure \ref{simp-ex} represented via (\texttt{P1 par (Q1;Q2)}), where the three possible execution interleavings are (\texttt{P1;Q1;Q2}), (\texttt{Q1;P1;Q2}) or (\texttt{Q1;Q2;P1}). Imagine the program here is executed under the first interleaving with the initial data store $\sigma$ where \texttt{X=1} and \texttt{Y=1}, resulting in the final state $\sigma'$ where \texttt{X=4} and \texttt{Y=6}. Without further information, inversion may assume the third interleaving was executed and so inverts the statements in the order (\texttt{P1;Q2;Q1}), clearly producing the incorrect final state where \texttt{X=4} and \texttt{Y=1}. We therefore will require both the auxiliary store $\delta$ and the inverse program as before, as well as the order in which the statements were executed forwards, termed \emph{interleaving order}.

Another challenge is the \emph{atomicity of statements}. Execution of a statement typically takes several steps to complete, increasing the possible execution paths and likelihood of races. Consider Figure \ref{simp-ex} with no assumption of atomicity and initial state $\sigma$ as above. Imagine \texttt{Y} is first evaluated in \texttt{P1}, then all of \texttt{Q1} and \texttt{Q2} are executed, before \texttt{P1} finally completes. This leads to the final state $\sigma'$ where \texttt{X=7} and \texttt{Y=3}, values not reachable when assuming statement atomicity.

Finally, \texttt{push} operations inserted in our first approach relate to a specific statement and these must be executed atomically. Consider the program \texttt{X+=1 par X=5} augmented via our first approach into \texttt{X+=1 par push($\sigma$(X),$\delta$(X)); X=5}, with no such atomicity and \texttt{X} initially \texttt{1}. Assume that the \texttt{push} statement executes first, storing the value \texttt{1} onto \texttt{$\delta$(X)}. The constructive assignment is then executed, incrementing \texttt{X} by one to \texttt{2}. Finally the destructive assignment is executed, updating \texttt{X} to \texttt{5}. The inverted version of this program is \texttt{X-=1 par X=pop($\delta$(X)}. Reversing the same interleaving first inverts the destructive assignment, assigning the value \texttt{1} from \texttt{$\delta$(X)} to the variable \texttt{X}. The constructive assignment is then inverted, decrementing \texttt{X} by one to \texttt{0}. Clearly, this reversal has been unsuccessful.

\subsection{Overcoming these Challenges} 
We update our approach to now capture the interleaving order. An identifier, or element of the set of natural numbers used in ascending order, is associated with each statement, each time it is executed, and stored onto a stack within the source code, very much like the \emph{communication keys} of CCSK \cite{IP2007,IP2012}. These identifiers index any reversal information stored, and are used to direct the execution of the inverted version. This makes the execution order deterministic, thus removing the first challenge. 

The updated approach will not introduce \texttt{push} statements in order to avoid issues relating to statement atomicity. A combination of this, and the fact that the interleaving order is not determined until runtime, mean all state-saving will be deferred to the operational semantics. A separate set of operational semantics are defined for forward execution. As such, inversion will no longer introduce \texttt{pop} statements, with all interaction with the reversal information being deferred to another separate set of operational semantics. In future work, we will add support for conditionals and loops, and further extend this with locks and mutual exclusion to allow this approach to be implemented within the language syntax.

We make the assumption of the atomicity of statements, restricting all interleaving to statement level, though this will be removed in future work.

\subsection{Annotation and Forward Execution} \label{sec:up-an-sos}
The process of augmentation is replaced with \emph{annotation}. This takes the original program, appends the necessary stacks into the program statements' source code, before returning the annotated version. Each statement is associated with a stack, necessary for programs containing loops as each statement may be executed multiple times requiring multiple identifiers. Stacks are not strictly necessary for our restricted language, however will be vital when we introduce conditionals and loops and so are included here for continuity. Each of these stacks is initially empty, and named uniquely via the function \texttt{nextS}. Let $\mathbb{S}'$ be the set of annotated statements. The function $ann: \mathbb{P} \rightarrow \hat{\mathbb{P}}$ takes the original program and recursively applies the re-defined function $a: \mathbb{S} \rightarrow \mathbb{S}'$ to each statement, producing the annotated version, where \texttt{e} is an arithmetic expression. 
\begin{alignat*}{2}
ann(\varepsilon) &= \varepsilon \texttt{ A} & \hspace{.7cm}  ann(\texttt{S;P}) &= a(\texttt{S}); ann(\texttt{P})  \\
a(\texttt{skip}) &= \texttt{skip A} & \hspace{.7cm} a(\texttt{X = e}) &= \texttt{X = e A}\\
a(\texttt{X cop e}) &= \texttt{X cop e A} & \hspace{.7cm}  a(\texttt{P par Q}) &= ann(\texttt{P}) \texttt{ par } ann(\texttt{Q}) 
\end{alignat*}
At this point, we have introduced a new syntactic category for annotated programs. Annotated programs and statements are defined below, with arithmetic and Boolean expressions as in Section \ref{sec:prog}. The sets $\mathbb{P}$ and $\mathbb{S}$ are also extended accordingly. 
\begin{align*}
\texttt{AP} &::= \varepsilon \texttt{ A} ~|~ \texttt{AS};~ \texttt{AP} \\
\begin{split} \texttt{AS} &::= \texttt{skip A} ~|~ \texttt{X = Exp A} ~|~ \texttt{X Cop Exp A} ~|~ \texttt{AP par AP} \end{split}
\end{align*}

Consider Figure \ref{simp-ex}. Applying the function $ann$ to this program produces the annotated version in Figure \ref{simp-ex-ann}, where each statement now has an empty stack. 

As mentioned previously, annotation does not handle state-saving with this now implemented within the operational semantics. Each time a statement execution completes, a unique identifier is added to that statement's source code stack. To synchronise the use of these identifiers, the next available identifier is retrieved through the function \texttt{next()}. To avoid further data races, there is mutual exclusion on the use of this atomic function between the parallel programs. This function, typically used as \texttt{m = next()}, assigns the value of the next available identifier to \texttt{m}, while incrementing the value it will return next time by one. Identifiers will index reversal information, meaning the stacks within the auxiliary store now consist of elements of the form (\texttt{i},\texttt{v}), where \texttt{i} is an identifier and \texttt{v} is a value. The following operational semantics defines the forwards execution, where $f($\texttt{A}$)$ indicates an update of the source code stack \texttt{A}. %

{\small \begin{align*}
&\text{[Skip]} \hspace{.5cm} \frac{}{(\texttt{skip A;P},\sigma,\delta) \rightarrow (\texttt{P},\sigma,\delta)} \hspace{.8cm} \text{[Seq1]} \hspace{.3cm} \frac{(\texttt{S A},\sigma,\delta) \rightarrow (\texttt{S$'$ $f($A$)$},\sigma',\delta')}{(\texttt{S A;P},\sigma,\delta) \rightarrow (\texttt{S$'$ $f($A$)$;P},\sigma',\delta')} \\[10pt]
\begin{split} &\text{[DA1]} \hspace{.5cm} \frac{}{(\texttt{X = v A},\sigma,\delta) \rightarrow (\texttt{skip m:A},\sigma[\texttt{X $\mapsto$ v}],\delta[\texttt{(m,$\sigma$(X)) $\mapsto$ X}])} \hspace{.2cm} \text{where } \texttt{m} = \texttt{next()} \end{split} \\[10pt]
\begin{split} &\text{[CA1]} \hspace{.5cm} \frac{}{(\texttt{X cop v A},\sigma,\delta) \rightarrow (\texttt{skip m:A},\sigma[\texttt{X $\mapsto$ $\sigma$(X) op v}],\delta)} \hspace{.5cm} \text{where } \texttt{m} = \texttt{next()} \end{split} \\[10pt]
&\text{[DA2]}  \hspace{.5cm} \frac{(\texttt{e},\sigma,\delta) \rightarrow (\texttt{e$'$},\sigma',\delta')}{(\texttt{X = e A},\sigma,\delta) \rightarrow (\texttt{X = e$'$ A},\sigma',\delta')} \hspace{.8cm} \text{[CA2]}  \hspace{.3cm} \frac{(\texttt{e},\sigma,\delta) \rightarrow (\texttt{e$'$},\sigma',\delta')}{(\texttt{X cop e A},\sigma,\delta) \rightarrow (\texttt{X cop e$'$ A},\sigma',\delta')} \\[10pt]
&\text{[P1]} \hspace{.5cm} \frac{}{(\texttt{P par skip},\sigma,\delta) \rightarrow (\texttt{P},\sigma,\delta)} \hspace{.8cm}
\text{[P2]} \hspace{.3cm} \frac{}{(\texttt{skip par Q},\sigma,\delta) \rightarrow (\texttt{Q},\sigma,\delta)} \\[10pt]
&\text{[P3]} \hspace{.5cm} \frac{(\texttt{P},\sigma,\delta) \rightarrow (\texttt{P$'$},\sigma',\delta')}{(\texttt{P par Q},\sigma,\delta) \rightarrow (\texttt{P$'$ par Q},\sigma',\delta')} \hspace{.8cm}
 \text{[P4]} \hspace{.3cm} \frac{(\texttt{Q},\sigma,\delta) \rightarrow (\texttt{Q$'$},\sigma',\delta')}{(\texttt{P par Q},\sigma,\delta) \rightarrow (\texttt{P par Q$'$},\sigma',\delta')}
\end{align*} }%
Sequential composition is handled similarly to our first approach, with the exception that the source code stacks are present and potentially modified. The expressions within assignments are handled either by [DA2] or [CA2] respectively, with no identifier association due to our assumption of the atomicity of statements. When the execution of a destructive assignment completes, the rule [DA1] associates a new identifier \texttt{m} within the source code stack \texttt{A}, and uses it to index the old value $\sigma$(\texttt{X}) stored on $\delta$. Constructive assignments complete via the rule [CA1], where an identifier is associated but no reversal information is stored. Parallel composition executes as expected, with either program able to make a step of execution, until one side is complete meaning the statement becomes sequential.

After the execution of the annotated program under these semantics, the \emph{final} annotated version with populated stacks is produced. Linking again to our example with the execution interleaving (\texttt{P1;Q1;Q2}), initial state $\sigma$ with values \texttt{X=1} and \texttt{Y=1} and initial auxiliary store $\delta$, the final annotated version is shown in Figure \ref{simp-ex-ann-final}. The program state $\sigma'$ after this execution has the values \texttt{X=4} and \texttt{Y=6}, while $\delta'$ contains the necessary reversal information.
\lstset{
numbers=none, 
numberstyle=\small, 
numbersep=2pt,
language=Pascal, 
framexleftmargin=15pt}
\begin{figure}[t]
  \begin{minipage}[b]{0.49\linewidth}
   \centering
    \begin{lstlisting}[xleftmargin=1.0ex,mathescape=true]
$\texttt{X}$+=$\texttt{Y}$+$\texttt{2}$[1] $\texttt{par}$ ($\texttt{Y}$=$\texttt{X}$+$\texttt{2}$[2]; $\texttt{X}$=$\texttt{4}$[3])
\end{lstlisting}
    \caption{Final Annotated program}
	\label{simp-ex-ann-final}
  \end{minipage}
  \hspace{0.01cm}
  \begin{minipage}[b]{0.49\linewidth}
    \centering
    \begin{lstlisting}[xleftmargin=1.0ex,mathescape=true]
$\texttt{X}$-=$\texttt{Y}$+$\texttt{2}$[1] $\texttt{par}$ ($\texttt{X}$=$\texttt{4}$[3]; $\texttt{Y}$=$\texttt{X}$+$\texttt{2}$[2])
	\end{lstlisting}
    \caption{Inverted program}
    \label{simp-ex-inv}
  \end{minipage}
\end{figure}

\subsection{Inversion and Reverse Execution}
Inversion now takes the \emph{final annotated program} and produces a relatively similar inverted version. This contains all statements of the given program in its inverted program order, with the inverted version of all constructive assignments. Due to the similarity between annotated and inverted versions, we now let $ \mathbb{S}'$ be the set of both annotated and inverted statements of this approach. The function $inv: \hat{\mathbb{P}} \rightarrow \mathbb{P}^{-1}$ recursively applies the re-defined function $i: \mathbb{S}' \rightarrow \mathbb{S}'$ to each statement in reverse order. Both $inv$ and $i$ are now given, with \texttt{icop} as defined in Section \ref{sec:inv}.
\begin{alignat*}{2}
inv(\varepsilon \texttt{ A}) &= \varepsilon \texttt{ A} & \hspace{.7cm}  inv(\texttt{AS;AP}) &= inv(\texttt{AP}); i(\texttt{AS})  \\
i(\texttt{skip A}) &= \texttt{skip A}  & \hspace{.7cm} i(\texttt{X = e A}) &= \texttt{X = e A} \\
i(\texttt{X cop e A}) &= \texttt{X icop e A} & \hspace{.7cm} i(\texttt{P par Q}) &= \texttt{(}inv(\texttt{P})\texttt{) par (} inv(\texttt{Q})\texttt{)} 
\end{alignat*}

The inverted version does not make use of the reversal information, and instead must be executed under a separate set of operational semantics for reverse execution. These semantics are responsible for all interaction with any information saved, as well as using the identifiers to direct inversion along the correct interleaving order. This is implemented using the mutually exclusive and atomic function \texttt{previous()}, related to the function \texttt{next()} such that \texttt{next()} = \texttt{previous() + 1}. The statement \texttt{m = previous()} checks that the current value of \texttt{m} matches the current value of \texttt{previous()}, as well as decrementing the value the function will return next time by 1. The statement \texttt{m == previous()} again checks that \texttt{m} is equal to \texttt{previous()}, but does not decrement the value it will return next time. This forces all steps of the evaluation to happen sequentially, reflecting our assumption of statement atomicity. The functions \texttt{previous()} and \texttt{next()} are strongly related, meaning the execution of one must update the value of the other accordingly. Here the rules for sequential and parallel composition are similar to those in Section \ref{sec:up-an-sos}, but with the transition relation $\rightsquigarrow$ replacing $\rightarrow$, hence they are omitted to save space.

{ \small \begin{align*}
\begin{split} &\text{[RDA]} \hspace{.5cm} \frac{\texttt{A = m:A$'$} \hspace{.5cm} \texttt{$\delta$(X) = (m,v):X$'$} \hspace{.5cm} \texttt{m = previous()}}{(\texttt{X = e A},\sigma,\delta) \rightsquigarrow (\texttt{skip A$'$},\sigma[\texttt{X $\mapsto$ v}],\delta[\texttt{X/X$'$}])}  \end{split} \\[7pt]
\begin{split} &\text{[RCA1]} \hspace{.5cm} \frac{\texttt{A = m:A$'$} \hspace{.5cm} \texttt{m = previous()}}{(\texttt{X cop v A},\sigma,\delta) \rightsquigarrow (\texttt{skip A$'$},\sigma[\texttt{X $\mapsto$ $\sigma$(X) op v}],\delta)}  \end{split} \\[7pt]
\begin{split} &\text{[RCA2]} \hspace{.5cm} \frac{(\texttt{e},\sigma,\delta) \rightsquigarrow (\texttt{e$'$},\sigma',\delta') \hspace{.5cm} \texttt{A = m:A$'$} \hspace{.5cm} \texttt{m == previous()}}{(\texttt{X cop e A},\sigma,\delta) \rightsquigarrow (\texttt{X cop e$'$ A},\sigma',\delta')}  \end{split} 
\end{align*} }%
Destructive assignments are handled via the single rule [RDA] as no evaluation of the expression \texttt{e} is required. A destructive assignment can be executed provided its most recent identifier matches both the current value of \texttt{previous()} and the index of the top element of its stack on $\delta$. Provided these conditions hold, the variable is restored to its previous value retrieved from its stack on $\delta$, before both of these stacks are popped. Constructive assignments require the two rules [RCA1] and [RCA2] as the expression must still be evaluated. Each step of the evaluation is executed sequentially by ensuring the identifiers match without removing them. Only when the assignment has executed will the identifiers be removed, restricting interleaving until this point. 

Applying the function $inv$ to the final annotated program in Figure \ref{simp-ex-ann-final} produces the inverted version in Figure \ref{simp-ex-inv}. Execution of this inverted version under the reverse operational semantics starting with the state $\sigma'$ with values \texttt{X=4} and \texttt{Y=6}, results in the reverse statement order of \texttt{Q2;Q1;P1}, the state $\sigma$ with values \texttt{X=1} and \texttt{Y=1} and the auxiliary store $\delta$. Therefore the execution has been successfully reversed with all variables restored to their initial values. 

\subsection{Correctness}
We now outline our correctness results for the second approach. Annotation of a program \texttt{P} assigns empty stacks to the statements of the program. During execution, these stacks are populated with identifiers. Let's denote such an \emph{update} of the stacks of $ann($\texttt{P}$)$ as $\rho(ann($\texttt{P}$))$. We now give propositions corresponding to those in Sections \ref{sec:aug} and \ref{sec:inv}, however we defer all termination parts to future work. Proposition \ref{prop-3} shows that the behaviour of the original and annotated programs are semantically equivalent with respect to the data store $\sigma$, and that the annotated program will populate both the stacks within the source code and the auxiliary store.
 
\begin{prop} \label{prop-3}
Let \textup{\texttt{P}} be a program and $ann($\textup{\texttt{P}}$)$ = \textup{\texttt{P$'$}}. If $(\textup{\texttt{P}},\sigma,\delta) \rightarrow^* (\textup{\texttt{skip}},\sigma',\delta)$ for some $\sigma'$, then $(\textup{\texttt{P$'$}},\sigma,\delta) \rightarrow^* (\textup{\texttt{skip C}},\sigma',\delta')$ for some \textup{\texttt{C}} and $\delta'$ and the computation $(\textup{\texttt{P$'$}},\sigma,\delta) \rightarrow^* (\textup{\texttt{skip C}},\sigma',\delta')$ produces an update $\rho($\textup{\texttt{P$'$}}$)$ for some $\rho$.
\end{prop}

Proposition \ref{prop-4} shows that executing the inverted program under the two stores and the updated source code stacks does indeed reverse all components to their initial values, as well as using the identifiers stored in the code to direct the execution.
\begin{prop} \label{prop-4}
Let \textup{\texttt{P}} be a program and $ann($\textup{\texttt{P}}$)$ = \textup{\texttt{P$'$}}. If $(\textup{\texttt{P}},\sigma,\delta) \rightarrow^* (\textup{\texttt{skip}},\sigma',\delta)$ for some $\sigma'$, then $(inv(\rho(\textup{\texttt{P$'$}})),\sigma',\delta') \rightsquigarrow^* (\textup{\texttt{skip C}},\sigma,\delta)$ for some \textup{\texttt{C}}, $\delta'$ and $\rho$.
\end{prop}

At least two additional lemmas are used throughout the proofs of the two propositions above. These correspond to the lemmas used in Sections \ref{sec:aug} and \ref{sec:inv}, and are listed below. 
\begin{lemma}
Let \textup{\texttt{S}} be a program statement and $ann($\textup{\texttt{S}}$)$ = \textup{\texttt{S A}} for some \textup{\texttt{A}}. If $(\textup{\texttt{S}},\sigma,\delta) \rightarrow^* (\textup{\texttt{skip}},\sigma',\delta)$ for some $\sigma'$, then $(\textup{\texttt{S A}},\sigma,\delta) \rightarrow^* (\textup{\texttt{skip A$'$}},\sigma',\delta')$ for some \textup{\texttt{A$'$}} and $\delta'$.
\end{lemma}

\begin{lemma}
Let \textup{\texttt{S}} be a program statement and $ann($\textup{\texttt{S}}$)$ = \textup{\texttt{S A}} for some \textup{\texttt{A}}. If $(\textup{\texttt{S}},\sigma,\delta) \rightarrow^* (\textup{\texttt{skip}},\sigma',\delta)$ for some $\sigma'$, then $(inv(\textup{\texttt{S A$'$}}),\sigma',\delta') \rightsquigarrow^* (\textup{\texttt{skip C}},\sigma,\delta)$ for some \textup{\texttt{A$'$}}, $\delta'$ and \textup{\texttt{C}}.
\end{lemma}

\section{Conclusion}
We have presented an approach to reversing an imperative programming language, using the state-saving notion. We have defined two functions, namely \emph{aug} and \emph{inv}, capable of producing the augmented version and inverted version of an originally irreversible program, respectively. We have proved that our augmentation does not alter the behaviour of the program with respect to the data store, and that it saves the necessary information to revert the program state after execution to that of before. The auxiliary store used to save this reversal information is also proved to revert to its initial state, ensuring no extra garbage data is produced.
  
We also described a modification to our first approach to include parallelism within a restricted language, while avoiding a number of issues parallelism introduces. We defined a function \emph{ann} and redefined \emph{inv} to support the recording of the interleaving order into the source code. Two sets of operational semantics are defined, one performing the state-saving for forwards execution, and another performing the inversion for reverse execution. Finally, we propose the correctness results for this modified approach. 

In the future, we shall relax the language restriction and support both conditional statements and while loops alongside parallel statements. The assumption of statement atomicity will be removed when considering a richer language which supports locks and mutual exclusion, allowing the approach described here to be implemented within the language itself. We will continue to extend the approach towards the complexity of C.

\subsubsection*{Acknowledgements}
We are grateful to the referees for their detailed and helpful comments and suggestions. The authors acknowledge partial support of COST Action IC1405 on Reversible Computation - extending horizons of computing. The second author acknowledges the support by the University of Leicester in granting him Academic Study Leave, and thanks Nagoya University for support during the study leave. The third author acknowledges the support by JSPS KAKENHI grants JP17H0722 and JP17K19969.

\bibliographystyle{eptcs}

\end{document}